\documentclass[11pt]{llncs}
\pdfoutput=1

\usepackage{geometry}
\usepackage{amsmath}
\usepackage{amsfonts}
\usepackage{bbm}
\usepackage{tikz}
 
\usepackage[ruled]{algorithm2e}

\SetArgSty{textrm}  % for algorithm2e
\SetAlFnt{\small}
\SetAlCapFnt{\small}
\SetAlCapNameFnt{\small}
\SetAlCapHSkip{0pt}
\IncMargin{-\parindent}

\newtheorem{observation}[theorem]{Observation}
\newtheorem{claim*}[theorem]{Claim}
\newcommand{\eps}{\epsilon}
\newcommand{\vals}{\mathbf{v}}

\newcommand{\OPT}{\operatorname{OPT}}
\newcommand{\LP}{\operatorname{LP}}
\def\RR{{\mathbb R}}
\def\cI{{\mathcal I}}
\def\E{{\mathbb E}}

\begin{document}
		
% Title portion
\title{When Are Welfare Guarantees Robust?%
\thanks{This work is supported by the I-CORE Program of the Planning and Budgeting Committee (PBC) and the Israel Science Foundation (grant number 4/11), and by NSF grants CCF-1215965 and CCF-1524062.}} 

\author{Tim Roughgarden
	\inst{1}
	Inbal Talgam-Cohen
	\inst{2}
	Jan Vondr\'ak
	\inst{1}}

\institute{Stanford University \email{tim@cs.stanford.edu} \email{jvondrak@stanford.edu} \and Hebrew University \email{inbal.talgamcohen@mail.huji.ac.il}}

\maketitle

\begin{abstract}
Computational and economic results suggest that social welfare maximization and combinatorial auction design are much easier when bidders' valuations satisfy the ``gross substitutes'' condition. The goal of this paper is to evaluate rigorously the folklore belief that the main take-aways from these results remain valid in settings where the gross substitutes condition holds only approximately. We show that for valuations that pointwise approximate a gross substitutes valuation (in fact even a {\em linear} valuation), optimal social welfare cannot be approximated to within a subpolynomial factor and demand oracles cannot be simulated using a subexponential number of value queries. We then provide several positive results by imposing additional structure on the valuations (beyond gross substitutes), using a more stringent notion of approximation, and/or using more powerful oracle access to the valuations.  For example, we prove that the performance of the greedy algorithm degrades gracefully for near-linear valuations with approximately decreasing marginal values; that with demand queries, approximate welfare guarantees for XOS valuations degrade gracefully for valuations that are pointwise close to XOS; and that the performance of the Kelso-Crawford auction degrades gracefully for valuations that are close to various subclasses of gross substitutes valuations.
\end{abstract}

%\keywords{Computational Game Theory, Mechanism Design}

% Section 1
\section{Introduction}
\label{sec:intro}

Welfare-maximization in combinatorial auctions is a central problem in
both theory and practice, and is perhaps the most
well-studied problem in algorithmic game theory (e.g.~\cite{BN07}).
The problem is: 
given a set $M$ of distinct items and
descriptions of, or oracle access to, the valuation functions
$v_1,\ldots,v_n$ of $n$ bidders (each a function from bundles to
values), determine the partition $S_1,\ldots,S_n$ of items that
maximizes the social welfare $\sum_{i=1}^n v_i(S_i)$.

Many possibility and impossibility results for efficiently computing
or approximating the maximum social welfare are known, as a function
of the set of allowable bidder valuations.  Very roughly, the current
state of affairs can be summarized by a trichotomy: (i) when the
valuations satisfy ``gross substitutes,'' then exact
welfare-maximization is easy; (ii) when the valuations are
``complement-free'' but not necessarily gross substitutes (e.g.,
submodular), exact
welfare-maximization is hard but constant-factor approximations are
possible; and (iii) with sufficiently general valuations, even
approximate welfare-maximization is hard.\footnote{``Gross
	substitutes'' states that a bidder's demand for an item can only
	increase as prices of other items are increased; see
	Section~\ref{sec:prelims} for a formal definition.}

The rule of thumb that ``substitutes are easy, complements are hard''
has its origins in the economic literature on multi-item auction
design.  For example, simple and natural ascending auctions converge
to a welfare-maximizing 
Walrasian equilibrium (i.e., to a market-clearing
price vector and allocation) whenever all bidders' valuations satisfy
gross substitutes, but when this condition is violated a Walrasian
equilibrium does not 
generally exist~\cite{KelsoCrawford1982,GS99,Mil00}.  Similarly, the VCG
mechanism has a number of desirable properties (like revenue
monotonicity) when bidders' valuations satisfy gross substitutes,
but not in general otherwise~\cite{AM02}.

Thus both computational and economic results suggest that
social welfare maximization and 
combinatorial auction design are much easier when
bidders' valuations satisfy the gross substitutes condition.
More generally, in settings with both substitutes and complements, the
folklore belief in the field is that simple auction formats should
produce allocations with near-optimal social welfare if and only if 
the substitutes component is in some sense ``dominant.''\footnote{For
	example, \cite{BCL00} write: ``In general, synergies across
		license valuations complicate the auction design process. Theory
		suggests
		that a `simple' (i.e., non-combinatorial) auction will have
		difficulty in assigning licenses efficiently in such
		an environment.''}
	For over twenty years, there has been a healthy (and high-stakes)
	debate over whether or not the ideal case of gross substitutes
	valuations should guide combinatorial auction design in
	realistic settings.\footnote{For example, \cite{ACMM97} write: ``A
			contentious issue in the design 
			of the Federal Communications Commission (FCC) auctions of
			personal communications services (PCS) licenses concerned the
			importance of synergies. If large synergies
			are prevalent among the licenses being offered, then the simultaneous
			ascending auction mechanism the
			FCC adopted, which does not permit all-or-nothing bids on sets of
			licenses, might be expected to perform
			poorly.''}
		
		{\em The goal of this paper is to evaluate rigorously the folklore belief
			that the main take-aways from the study of gross substitutes
			valuations remain valid in settings where the condition holds only
			approximately.}
		
		The motivation for this work is further strengthened by the recent interest in \emph{data-driven optimization} \cite{SV15,HS16,BRS16}. In general, in order to optimize an objective (in our case social welfare) in a real-world setting, a model must first be formulated, and its parameters (in our case valuations) must be learned from the available data. The estimated model is only an approximation of the true setting, and may exhibit arbitrary errors introduced by the learning techniques. 
		For this reason, there has been much recent interest in the learning community in optimization of functions that can be only be evaluated approximately \cite{BLNR15,CDDK15}, and of functions for which ``nice'' properties such as submodularity hold only approximately \cite{KC10}.%
		\footnote{An orthogonal line of research studies \emph{approximately learning} functions with ``nice'' properties \cite{BH11}.}
		By the same token, as market and mechanism design become more data-driven, it is important to understand the possibilities and impossibilities of welfare optimization under estimation errors that cause the gross substitutes condition to hold only approximately.
		
		\subsection{Our Results}
		
		We first consider arguably the most natural notion of ``approximate
		gross substitute valuations,'' namely valuations that are pointwise
		within a $1+\eps$ factor of some gross substitutes valuation.  Do
		the laudable properties of gross substitutes valuations degrade
		gracefully as $\eps$ increases?  At this level of generality, the
		answer is negative: even for valuations that are pointwise close to {\em
			linear} valuations, we prove that the social welfare cannot be
		approximated to within a subpolynomial factor using a subexponential
		number of value queries, and that demand oracles cannot be
		approximated in any useful sense by a subexponential number of value
		queries.  When the gross substitutes condition holds exactly, a demand
		query can be implemented using a polynomial number of value queries
		(see \cite{Pae14}),
		and the welfare-maximization problem can be solved exactly with a
		polynomial number of such queries (see \cite{NS06}).
		We conclude that there is no sweeping generalization of the properties
		of gross substitutes valuations to approximations of such valuations,
		and that any positive result must impose additional structure on the
		valuations (beyond gross substitutes), use a more stringent
		notion of approximation, and/or use more powerful oracle access to the
		valuations.
		
		We next consider positive results in the value oracle model.  Our main
		result here is a proof that the (optimal) performance of the greedy
		algorithm degrades gracefully for valuations that are close to linear
		functions, provided these valuations are also approximately
		submodular (in a stronger than pointwise sense).  (As noted above,
		some assumption beyond near-linearity is necessary for any positive result.)
		The standard arguments for proving approximation bounds for the
		greedy algorithm (e.g.~\cite{LLN06}) do not imply this result,
		and we develop a new analysis for this purpose.
		
		We then consider welfare-maximization with demand oracles,\footnote{As
			noted above, with valuations that only satisfy gross substitutes
			approximately, demand oracles are substantially more powerful than
			value oracles.}  and find that welfare guarantees tend to be more
		robust in this model. First, the {\em value} of optimal social welfare
		can be approximated using demand queries within a factor of
		$\gamma_{\cal C} + \epsilon$, whenever the valuations are pointwise
		$\epsilon$-close to a class $\cal C$ such that the integrality gap of
		the ``configuration LP'' is $\gamma_{\cal C}$. Since the configuration
		LP is the primary vehicle for developing approximation algorithms in
		the demand oracle model, we recover pointwise robustness essentially
		for all known approximation results in the demand oracle model (in
		terms of the optimal value).  Another question, though, is whether an
		allocation achieving good welfare can be found in a
		computationally-efficient way. 
		For some classes of valuations, we show
		that this is possible (and thus we achieve a
		$(1-\epsilon)$-approximation for valuations $\epsilon$-close to
		linear, and a $(1-1/e-\epsilon)$-approximation for valuations
		$\epsilon$-close to XOS). We remark that no extra assumption of
		near-submodularity is required here; this highlights another
		difference between the value and demand oracle models. 		
		
		Another approach to finding an optimal allocation assuming the gross substitutes
		property is the Kelso-Crawford algorithm, also known as the
		{\em t\^atonnement} procedure. In general, this procedure requires demand
		queries and thus also falls within the umbrella of the demand oracle model.
		Here we show that the performance of the Kelso-Crawford algorithm degrades
		smoothly for certain classes of functions, namely for valuations close to linear,
		close to unit-demand with $\{0,1\}$ values, and more generally close to transversal valuations
		(rank functions of a partition matroid). 
		Unfortunately the Kelso-Crawford algorithm is not going to be a universally
		robust solution for general gross substitutes, either. As we show, its performance
		degrades discontinuously for valuations close to unit-demand (with unrestricted values),
		a seemingly minor extension of the cases above where we showed positive
		results. 
		
		In addition we show a counterexample to an approach based on Murota's cycle
		canceling algorithm, 
		the remaining known way to solve the welfare-maximization problem
		under the gross substitutes property.
		
		In summary, the idea
		that approximation guarantees for various classes
		of valuations should degrade gracefully under small deviations from
		the class 
		should be viewed with skepticism.
		Our negative results do not imply that that the folklore belief that ``close to
		substitutes is easy'' is wrong, but they do imply that there is no
		``generic reason'' for why this might be the case.
		Our positive results show that robust guarantees can still be obtained
		in certain cases, but typically this requires additional care and
		often new ideas on top of known techniques.

\subsection{Related Work and Organization}
\label{sub:related-work}

\emph{Valuation functions.} There are several increasingly-restrictive classes of valuations studied in the literature that may be considered ``free from complements'', including subadditive, XOS, submodular and gross substitutes valuations. Beginning with \cite{LLN06},
the algorithmic game theory literature has extensively studied
submodular valuations and the hierarchy of their superclasses;
recently there has been growing interest in the gross substitutes
subclass as well \cite{Pae14,HMR+16,PW16}. 

\emph{Welfare maximization.} The study of welfare maximization for gross substitutes has a deep mathematical and algorithmic basis. Even the case of unit-demand valuations subsumes the bipartite matching problem, for which classic combinatorial and linear programming techniques were developed (see, e.g., \cite{PS00}). 
In the general case, the main techniques for gross substitutes include
ascending-price auctions \cite{KelsoCrawford1982}, combinatorial cycle
canceling algorithms based on discrete convexity \cite{Mur96a,Mur96b},
and linear programming \cite{BM97,NS06} (see also \cite{Pae14,PW16}). 

For submodular valuations, exact welfare maximization is hard and so research efforts have concentrated on approximation. A 2-approximation can be achieved greedily \cite{LLN06}, and
the best approximation algorithm in the value
oracle model (continuous greedy) achieves a
$(1-1/e)$-approximation \cite{Von08}, which is optimal
\cite{KLMM08}. In the demand query model, a 2-approximation can be
achieved via an ascending-price auction \cite{FKL12}, and the $1-1/e$
barrier can be broken \cite{FV10}.  

\emph{Welfare maximization under deviations.} Not much previous work has studied the extent to which 
approximate welfare-maximization guarantees 
approximately hold under small deviations from the valuation class.
The closest result to our work is a theorem of \cite{HS16},
who show that one cannot achieve a constant-factor
approximation when maximizing
valuations that are pointwise $\epsilon$-close to submodular
(in the sense of our Definition~\ref{d:close}).\footnote{The model in
	\cite{HS16} is that of {\em erroneous value oracles},
	which return values $\tilde{v}(S)$ close to an underlying submodular
	function $v(S)$. This is equivalent to our Definition~\ref{d:close},
	as we might as well assume that $\tilde{v}(S)$ is the true
	valuation.}  Our Proposition~\ref{prop:close-to-linear-hard} 
strengthens this negative result to apply even to valuations
that are close to linear. For related work on information-theoretic lower bounds (in the context of submodular optimization) see, e.g., \cite{MSV08,PSS08}.

We now survey other related work, that
differs from ours in the model of deviation, the valuation classes
considered, and/or the algorithmic questions it focuses on.   
Hassidim and Singer \cite{HS16} also study the problem of welfare maximization with
submodular valuations and random (rather than adversarial) noise. Milgrom
\cite{Mil15} introduces a metric for distance from a matroid rank
function, in terms of distance to the matroid's linear constraints;
the algorithmic setting is different from ours since the players are
single-parameter subject to a single feasibility constraint, rather
than part of a combinatorial auction setting. Feige et al.~\cite{FFI+14} study, and unify previous results on, deviations from
submodularity that are captured by a graphical representation. Maehara and Murota
\cite{MM15} adopt a heuristic approach (not guaranteed to converge)
based on Murota's cycle canceling algorithm, and test how well it
performs in experiments for subclasses of submodular valuations.  

It is interesting to compare our work to that of \cite{KD07}, who
study divisible items with budgets rather than indivisible items with quasi-linear utilities. They define an
\emph{approximate weak gross substitutes} property for which the Garg-Kapoor auction
ensures that each player approximately gets his demand (in a different sense than our
notion of approximate demand, as formulated below in Definition
\ref{def:WE-approx}, Inequality \eqref{eq:WE-approx}). 
Their conclusion that ``markets do not suddenly become intractable if
they  slightly violate the weak gross substitutes property'' is
therefore quite different from ours, showing an intriguing discrepancy
between the models.  

\subsubsection{Organization.}

After introducing the preliminaries in Section \ref{sec:prelims}, Section \ref{sec:cautionary-tales} presents two cautionary tales as to what can go wrong when simple valuations undergo smalls perturbations. We then establish positive results in the value oracle model (Section \ref{sec:results-value}) and in the demand oracle model (Section \ref{sec:results-demand}), and conclude
with some open questions in Section \ref{sec:open-questions}. We defer some proofs to
Appendices~\ref{appx:welfare-algos}--\ref{appx:KC}, as well as our negative examples for specific algorithms (Appendix~\ref{appx:negative}).

% Section 2
\section{Preliminaries}
\label{sec:prelims}

\subsubsection{Combinatorial Auctions: Model and Notation.}

A \emph{combinatorial auction} setting (or \emph{market}) includes a set of $n$ players $N$ and a set of $m$ indivisible items $M$. 
We call a subset $S\subseteq M$ of items a \emph{bundle}. For an ordered set of items $S=(s_1,s_2,\dots,s_k)$ and for $j\in[k]$, we use the notation $S_j$ to denote the ``prefix'' subset $\{s_1,\dots s_{j-1}\}$ of items that appear before $s_j$ (in particular, $S_1=\emptyset$).

Every player $i$ has a \emph{valuation} $v_i:2^M \rightarrow \RR_+$ over bundles.
Every valuation $v$ is assumed to be monotone unless stated otherwise, i.e., $v(S)\le v(T)$ for every two bundles $S\subseteq T$. Valuations are not necessarily normalized, i.e., $v(\emptyset)$ is not always equal to zero (normalization is not without loss of generality in our model of perturbation). For every two bundles $S,T$ we denote by $v(S\mid T)$ the \emph{marginal} value of $S$ given $T$, which is equal to $V(S\cup T)-v(T)$. 

An \emph{allocation} $\mathcal{S}$ is a (possibly partial) partition of the items in $M$ into $n$ bundles $S_1,\dots,S_n$, where bundle $S_i$ is the allocation of player $i$. A \emph{full} allocation is an allocation in which every item is allocated. The social efficiency of an allocation $\cal S$ is measured by its \emph{welfare} $W(\mathcal S)=\sum_{i=1}^n
v_i(S_i)$.

A \emph{price vector} $p\in \RR_+^m$ is a vector of item prices, one per item. For every bundle $S$ we denote by $p(S)$ the aggregate price $\sum_{j\in S}p(j)$ of the bundle.
Given a price vector $p$, each player wishes to maximize his \emph{quasi-linear utility}, i.e., to receive a bundle $S$ maximizing $v_i(S)-p(S)$ among all possible bundles (including the empty one). A utility-maximizing bundle is said to be \emph{in demand} for the player given the price vector $p$.
If the player is already allocated a bundle $T$, then he wishes to add to his allocation a bundle $S\subseteq (M\setminus T)$ maximizing $v_i(S\mid T)-p(S)$, and we say that $S$ is in his demand given $p$ and $T$.
A bundle $S$ is \emph{individually rational (IR)} for a player $i$ if $v_i(S)\ge p(S)$; it is \emph{strongly IR} if every subset $T\subseteq S$ is IR for $i$. 

A full allocation $\mathcal S$ and price vector $p$ form a \emph{Walrasian equilibrium} if for every player $i$, $S_i$ is in $i$'s demand given $p$. By the \emph{first welfare theorem}, the equilibrium allocation $\mathcal S$ maximizes welfare.

Since valuations are in general of exponential size in $m$, the standard assumption is that they are accessed via \emph{value oracles}, which return the value of any bundle upon query. In Sections \ref{sub:value-vs-demand} and \ref{sec:results-demand} we discuss \emph{demand oracles}, which given a price vector $p$, return a bundle in the player's demand given $p$. In this case we allow $p$ to include negative prices so that a demand oracle will also be able to return a bundle in the player's demand \emph{given a previous allocation} $T$.

\subsubsection{Valuation Classes.}

We define several valuation classes of interest.
A valuation $\ell$ is \emph{linear} if there exists a vector $(l_1,\dots,l_m)\geq 0$ and a scalar $c\ge 0$ such that $\ell(S) = c + \sum_{j \in S} l_j$ for every bundle $S$; it is \emph{additive} if $c=0$.
A valuation $r$ is \emph{unit-demand} if there exists a vector $(\rho_1,\dots,\rho_m)\geq 0$ such that $r(S) = \max_{j \in S} \rho_j$ for every bundle $S$.

Let $\mathcal M=(M,\cI)$ be a matroid over the ground set of items $M$, where $\cI$ is the family of \emph{feasible} bundles.%
\footnote{A set system $\mathcal M=(M,\cI)$ is a matroid if the following properties hold: (i) $\cI$ is non-empty; (ii) $\cI$ is downward-closed, that is, if $S\subseteq T$ and $T\in\cI$ then $S\in\cI$; (iii) for every $S,T\in\cI$ such that $|S|<|T|$, there exists $t\in T\setminus S$ such that $S\cup\{t\}\in\cI$ (see, e.g., \cite{Ox-92}).}
A maximal feasible set is called a \emph{basis}, and the matroid \emph{rank function} maps bundles to the cardinality of their largest feasible subset. Given a vector $(w_1,\dots,w_m)\ge 0$ of item weights, the corresponding \emph{weighted} rank function maps bundles to the weight of their heaviest feasible subset.  
A valuation $r$ is an unweighted (resp., weighted) matroid rank function if there exists a matroid $\mathcal M$ such that $r$ is its (weighted) rank function. 
Both linear and unit-demand valuations are types of weighted matroid rank functions.

A valuation $v$ is \emph{gross substitutes} if for every two price vectors $p\le q$, for every bundle $S$ in demand given $p$, there exists a bundle $T$ in demand given $q$, which contains every item $j\in S$ for which $q(j)=p(j)$.%
\footnote{The intuition behind this definition will become clearer once we define the Kelso-Crawford algorithm below.}
A valuation $v$ is \emph{submodular} if for every two bundles $S\subseteq T$ and item $j\notin T$, $v(j\mid S)\ge v(j\mid T)$, i.e., the marginal value of $j$ is decreasing. A valuation $v$ is \emph{XOS} (also known as \emph{fractionally subadditive}) if there exist additive valuations $a_1,a_2,\dots$ such that $v(S)=\max_k a_k(S)$ for every bundle $S$. 
It is well-known that gross substitutes $\subset$ submodular $\subset$ XOS.

\subsubsection{Algorithmic Approaches to Welfare-Maximization.}

Finding a welfare-maximizing allocation in combinatorial auctions is computationally tractable for gross substitutes. There are three main algorithmic approaches: (1) auction-based, (2) LP-based, and (3) combinatorial. For linear valuations there is also (4) the greedy approach. The first two approaches use demand queries and run in polynomial time, while the latter two use value queries and run in strongly-polynomial time. 
We describe here Approaches (1) and (4) in more detail; see also Appendix \ref{appx:welfare-algos}. 

In each round of the Kelso-Crawford algorithm, one player (chosen arbitrarily) adds to his existing allocation a bundle in his demand, taking into account his current allocation (and a $\delta$-increase in the prices of items not currently in his allocation). The algorithm terminates only when no player wants to add to his current allocation.
For completeness, a full description appears in Algorithm \ref{alg:KC} in Appendix \ref{appx:welfare-algos}.

\begin{proposition}[\cite{KelsoCrawford1982}]
	\label{pro:tatonn}
	For gross substitutes valuations, Algorithm \ref{alg:KC} converges to a (welfare-maximizing) Walrasian equilibrium as $\delta\to 0$.  
\end{proposition}

The proof of Proposition \ref{pro:tatonn} relies on the definition of gross substitutes, by which at any point in the algorithm, every player's allocation is a subset of some bundle in his demand.

Algorithm \ref{alg:greedy} in Appendix \ref{appx:welfare-algos} formally describes the greedy algorithm for finding a value-maximizing bundle subject to a feasibility constraint $\cI$ on the bundles; this is a problem to which welfare-maximization is well-known to reduce (see Proof of Corollary \ref{cor:greedy-linear}).
$\cI$ is simply a family of feasible bundles.% 
\footnote{Since it can be exponential in size we assume it is given by an oracle which returns whether a bundle is feasible or not.} 
The greedy algorithm starts with an empty bundle, and in each iteration adds the item with maximum marginal value among all items that maintain feasibility, breaking ties arbitrarily. 

\begin{proposition}[\cite{FNW78}]
	If $v$ is linear and the feasibility constraint $\mathcal I$ is such that the set system $\mathcal M = (M,\cI)$ is a matroid, then Algorithm \ref{alg:greedy} is optimal.
\end{proposition}

% Section 3
\section{Two Cautionary Tales}
\label{sec:cautionary-tales}

Do the nice properties of a valuation class degrade gracefully as one
moves outside the class?
This section describes two cautionary tales 
demonstrating that the answer can subtly depend on the notion of
``being  close,'' on the tractable class under consideration, and on
the model of access to the valuations.	 

Arguably the most natural general notion of ``closeness'' is
pointwise.\footnote{We use relative error in
	the interests of scale-invariance (the "units" in which valuations are specified do not matter). Absolute additive error is not interesting for the problems that we study.}
\begin{definition}\label{d:close}
	A valuation $\tilde{v}$ is \emph{$\epsilon$-close} to
	submodular (or linear, gross substitutes, etc.), if there is a
	submodular (or linear, gross substitutes, etc.) valuation $v$ such
	that $v(S) \leq \tilde{v}(S) \leq (1+\epsilon) v(S)$ for all bundles
	$S$. 
\end{definition}

\subsection{Close-to-additive vs. Close-to-linear Valuations}
\label{sub:basic-hardness}

We next show that approximate welfare-maximization is hard for
valuations that are $\eps$-close to linear (a restrictive subclass of
gross substitutes valuations). 

\begin{proposition}
	\label{prop:close-to-linear-hard}
	Given value oracles for valuations $\epsilon$-close to linear, no
	algorithm using a subexponential number of queries can approximate the
	value of optimal social welfare within a factor better than polynomial
	in $m,n$. 
\end{proposition}

\begin{proof}
	Let $|M| = m = a n$ and let $(A_1, A_2, \ldots, A_n)$ be a random
	partition of $M$ such that $|A_i| = a$. 
	We define linear valuations as follows:
	$$ \ell_i(S) = \epsilon + \frac{1}{a} |S \cap A_i|.$$
	We also define
	$$ w_i(S) = (1+\epsilon) \epsilon + \frac{1}{an} |S|.$$
	Note that if we do not know the partition $(A_1,\ldots,A_n)$, which is
	randomized, $\ell_i(S)$ will be close to its expectation which is
	$w_i(S)$.  
	Let us define the following valuations $v_i$:
	\begin{itemize}
		\item If $\big| |S \cap A_i| - \frac{1}{n} |S| \big| \leq \epsilon^2 a$, then $v_i(S) = w_i(S)$.
		\item If $\big| |S \cap A_i| - \frac{1}{n} |S| \big| > \epsilon^2 a$, then $v_i(S) = \ell_i(S)$.
	\end{itemize}
	Note that in the first case, we have $v_i(S) \geq (1+\epsilon)
	\epsilon +  \frac{1}{a} (|S \cap A_i| - \epsilon^2 a) = \ell_i(S)$ and
	$v_i(S) \leq (1+\epsilon) \epsilon + \frac{1}{a} (|S \cap A_i| +
	\epsilon^2 a) \leq (1+2\epsilon) \epsilon + \frac{1}{a} |S \cap A_i|
	\leq (1+2\epsilon) \ell_i(S)$. So $v_i$ is $2\epsilon$-close to
	$\ell_i$. 
	
	By Chernoff bounds, for a fixed query $S$, the probability that $\big|
	|S \cap A_i| - \frac{1}{n} |S| \big| > \epsilon^2 a$ is
	$e^{-\Omega(\epsilon^4 a)}$. Hence, with high probability we are
	always in the first case above, the returned value depends only on
	$|S|$ and hence the algorithm does not learn any information about
	$A_i$. Therefore (by standard arguments), it would require
	exponentially many queries to find any set such that $\big| |S \cap
	A_i| - \frac{1}{n} |S| \big| > \epsilon a$ and hence distinguish
	whether the input valuations are $v_i$ or $w_i$. 
	
	The optimal solution under $v_i$ is $S_i = A_i$ which gives welfare
	$\sum_{i=1}^{n} v_i(A_i) > \frac{1}{a} \sum_{i=1}^{n} |A_i| = n$,
	while the optimal welfare under $w_i$ is $(1+\epsilon) \epsilon n +
	1$. 
	So the approximation factor cannot be better than $(1+\epsilon)
	\epsilon + \frac{1}{n}$. 
	
	Note that we can set $\epsilon = 1/a^{1/4-\delta} =
	(n/m)^{1/4-\delta}$ and the high probability statements still
	hold. Therefore, we can push the hardness factor to $\max \{
	(n/m)^{1/4-\delta}, 1/n \}$. 
	\qed
\end{proof}

This hardness result relies strongly on the value oracle
model --- things are quite different in the demand oracle model, as we
discuss  in Section~\ref{sub:value-vs-demand} below. 

Proposition~\ref{prop:close-to-linear-hard} is a startling contrast to
the case of valuations that are $\eps$-close to {\em additive}
valuations, which differ only by requiring the empty set
to have value $0$. 
Here, approximate welfare-maximization is easy.

\begin{proposition}\label{prop:additive}
	Welfare-maximization can be solved within a factor of $1+\epsilon$ for
	$\epsilon$-close to additive valuations. 
\end{proposition}

\begin{proof}
	An additive valuation has the form $v(S) = \sum_{j \in S} a_j$.
	By assumption, $a_j \leq v(\{j\}) \leq (1+\epsilon) a_j$ for every
	singleton  $j$. Hence we can determine the coefficients $a_j$ within a
	factor of $1+\epsilon$ and run welfare maximization on the resulting
	additive functions. Each item simply goes to the highest bidder, and
	we lose a factor of at most $1+\epsilon$ due to the errors in
	determining $a_j$. 
	\qed
\end{proof}

The proof of Proposition~\ref{prop:additive} makes the more general
point that, whenever the approximating ``nice'' valuation $\tilde{v}$
can be (approximately) recovered from the given valuation $v$, then
welfare approximation guarantees carry over (applying an off-the-shelf
approximation algorithm to the valuations $\tilde{v}$).  As
Proposition~\ref{prop:close-to-linear-hard} makes clear, however, in
many cases it is not possible to efficiently reconstruct an
approximating ``nice'' valuation, even under the promise that such a
valuation exists.

\subsection{Value Queries vs. Demand Queries}
\label{sub:value-vs-demand}

One remarkable property of gross substitutes valuations is
that there is no difference between the value oracle and demand oracle
models: Demand queries can be simulated efficiently by value queries
(via the greedy  algorithm), and
hence any algorithm in the demand oracle model can also be implemented
in the value oracle model (\cite{Pae14} and references within). 
Unfortunately, this is no longer true for valuations that are 
$\eps$-close to gross substitutes, or even $\eps$-close to linear.

\begin{proposition}
	\label{prop:demand-hard}
	Given a value oracle to a valuation $\epsilon$-close to linear,
	answering  demand queries requires an exponential number of value
	queries. 
\end{proposition}

This can be proved by a construction similar to the one
above. However, let us instead present an indirect argument which
shows even a  little bit more.  

\begin{lemma}
	\label{lem:LP-val}
	For any class of valuations $\cal C$ such that the integrality gap of
	the configuration LP is $\gamma\ge 1$, it is possible to estimate the
	optimal social welfare within a multiplicative factor of $(1+\epsilon)
	\gamma$ for valuations that are $\epsilon$-close to $\cal C$ in the demand
	oracle model. 
\end{lemma}

\begin{proof}
	Consider the configuration LP:
	\begin{eqnarray*}
		\max & \sum_{i=1}^{n} \sum_{S \subseteq [m]} v_i(S) x_{i,S}: \\
		\forall j \in [m];  & \sum_{i=1}^{n} \sum_{S: j \in S} x_{i,S} \leq 1, \\
		\forall i \in [n]; & \sum_{S \subseteq [m]} x_{i,S} = 1, \\
		x_{i,S} \geq 0.
	\end{eqnarray*}
	Let us denote by $\LP(\vals)$ and $\OPT(\vals)$ the LP optimum and optimal
	social welfare, respectively, under valuations $\vals$. The assumption
	is that for valuations $\vals \in \cal C$, $\OPT(\vals) \leq \LP(\vals) \leq
	\gamma \cdot \OPT(\vals)$. Let us solve the LP for valuations
	$\tilde{\vals}$ that are pointwise $\epsilon$-close to valuations $\vals
	\in \cal C$. (This is possible since we assume that we have demand
	oracles for $\tilde{\vals}$, which gives a separation oracle for the
	dual; see~\cite{NS06}.) We have $v_i(S) \leq \tilde{v}_i(S) \leq (1+\epsilon) v_i(S)$
	for each $i$ and $S$. Therefore, 
	the same inequalities hold for the LP optimum as well as optimal
	social welfare, and we obtain 
	$$ \OPT(\tilde{\vals}) \leq \LP(\tilde{\vals}) \leq (1+\epsilon) \LP(\vals)
	\leq (1+\epsilon) \gamma \cdot \OPT (\vals) \leq (1+\epsilon) \gamma
	\cdot \OPT(\tilde{\vals}).$$ 
\end{proof}

Since the configuration LP is known to be integral for gross
substitutes valuations (see e.g.~\cite{Voh11}), we obtain the
following. 

\begin{corollary}
	\label{cor:LP-val}
	For valuations $\epsilon$-close to gross substitutes, it is possible
	to estimate the optimal social welfare to within a multiplicative
	factor of $1+\epsilon$ in the demand oracle model.  
\end{corollary}

This implies Proposition~\ref{prop:demand-hard}: a simulation of
demand queries would allow us to approximate social welfare within
$1+\epsilon$ for valuations $\epsilon$-close to gross substitutes and
as a special case $\epsilon$-close to linear. 
We know from Proposition~\ref{prop:close-to-linear-hard} that this
(and even much weaker approximations) would require  exponentially
many value queries.  Actually we obtain a stronger statement:
It is not possible to answer demand queries via value queries for
$\epsilon$-close to linear valuations even approximately, under {\em
	any notion of approximation} that would be useful for approximating
the optimal social welfare. This is because, again, such a simulation
would lead to a $(1+\epsilon)$-approximation of social welfare via
value queries, which Proposition~\ref{prop:close-to-linear-hard} rules
out. 

% Section 4
\section{Positive Results for Restricted Valuation Classes}
\label{sec:results-value}

Motivated by the hardness results in the previous section for valuations 
$\eps$-close to linear
in the value oracle model, this section considers stronger notions of
closeness. In Section \ref{sub:closeness-measures} we discuss ``marginal''
closeness including the notion  of \emph{$\alpha$-submodularity}. In Section
\ref{sub:greedy-pos} we analyze the greedy algorithm for valuations
that are $\alpha$-submodular in addition to being $\epsilon$-close to
linear.

\subsection{Marginal Closeness}
\label{sub:closeness-measures}

Recall that a pointwise approximation approximates the outputs of a valuation as a discrete function. 
A natural way to strengthen this is to approximate the behavior of the valuation's discrete derivatives, or marginal values. The marginal value $v(j\mid\cdot)$ of every item $j$ can be viewed as a set function, and together these set functions encode important properties of the valuation. 
For example, linear valuations can be characterized as valuations for which the marginal of every item is a constant function; 
submodular valuations can be characterized as valuations for which the marginal of every item is a non-increasing function;%
\footnote{Since for every $j$ and $S\subseteq T$, $v(j\mid S)\ge v(j\mid T)$.} 
and unweighted matroid rank functions can be characterized as valuations for which the marginal of every item is a non-increasing zero-one function \cite{Sch03}. 

We remark that even without any strengthening, pointwise $\epsilon$-closeness has the following useful implication regarding the ``closeness'' of the marginal values of bundles:

\begin{observation}
	\label{obs:marginals_close}
	For every valuation $v$ that is $\epsilon$-close to a valuation $v'$, and for every two bundles $S,T\subseteq M$, 
	\begin{equation}
	v'(S\mid T) - \epsilon v'(T) \le v(S\mid T) \le v'(S\mid T) + \epsilon v'(S\cup T).\label{eq:marginals_close}
	\end{equation}
\end{observation}

\begin{proof}
	By the definition of $\epsilon$-closeness, $v(S\mid T)=v(S\cup T)-v(T) \ge v'(S\cup T)-(1+\epsilon)v'(T)=v'(S\mid T)-\epsilon v'(T)$.
	The upper bound follows similarly.
	\qed
\end{proof}

Observation \ref{obs:marginals_close} will be useful in proving our positive results in Sections \ref{sub:greedy-pos} and \ref{sub:KC-pos}. Yet the guarantee in \eqref{eq:marginals_close} is relatively weak: the ``error terms'' depend on $v'(T)$ or $v'(S\cup T)$, and so can be large if these terms are large. 
We now define the stronger notion of marginal pointwise approximation:%
\footnote{Equivalent to having erroneous access to the marginal values via an oracle \cite{CCPV11}.} 

\begin{definition}
	\label{def:marginal-close}
	A valuation $v$ is \emph{marginal-$\epsilon$-close} to a class $\mathcal C$ of marginals (constant, non-increasing, etc.), if for every item $j$ the corresponding marginal value function $v(j\mid\cdot)$ is (pointwise) $\epsilon$-close to a function $g_j\in\mathcal C$. 
\end{definition} 

Let us now see what properties follow from this definition. 

\subsubsection{The Property of $\alpha$-Submodularity.}

Consider first the class of valuations that are marginal-$\epsilon$-close to decreasing. 
This coincides with the previously-studied class of $\alpha$-submodular valuations -- a strict super-class of submodular valuations:

\begin{definition}[\cite{LLN06}]
	\label{def:alpha-submod}
	Let $\alpha\ge 1$. 
	A valuation $v$ is \emph{$\alpha$-submodular} if for every two bundles $S\subseteq T$ and item $j\notin T$, $\alpha v(j\mid S)\ge v(j\mid T)$.
\end{definition}

\begin{observation}
	A valuation $v$ is marginal-$\epsilon$-close to the class of decreasing functions if and only if it is $(1+\epsilon)$-submodular.
\end{observation}

\begin{proof}
	For every $S\subseteq T$ and $j\notin T$, let $g_j$ be the decreasing function to which $v(j|\cdot)$ is $\epsilon$-close. Then $(1+\epsilon)v(j|S)\ge (1+\epsilon)g_j(S)\ge (1+\epsilon)g_j(T) \ge v(j|T)$, showing that $v$ is $(1+\epsilon)$-submodular.
	The converse follows from
	Observation \ref{obs:alpha-submod}
	in Appendix \ref{appx:alpha-submod}.
	\qed 
\end{proof}

Lehmann et al. \cite{LLN06} show that valuations which are $\alpha$-submodular have welfare guarantees similar to those of submodular valuations. This is in stark contrast to the hardness of sub-polynomial approximation for valuations that are only pointwise $\epsilon$-close to submodular \cite{HS16}.

Our positive result in Section \ref{sub:greedy-pos} is for $\epsilon$-close to linear valuations that are also $\alpha$-submodular. Note that the submodularity assumption is a standard one in the context of algorithmic game theory (as justified by \cite{LLN06}), and that $\alpha$-submodularity is a strictly weaker assumption.
We next show that $\epsilon$-close to linear, $\alpha$-submodular valuations subsume an interesting subclass. 

\subsubsection{An Interesting Special Case: Bounded Curvature.}

Consider the class of valuations that are marginal-$\epsilon$-close to
the constant (marginal) functions. It is not hard to see that this includes the class of
submodular valuations with \emph{bounded curvature}. 

\begin{definition}[\cite{CC84,SVW15}]
	A valuation $v$ has curvature $c\in[0,1]$ if for every item $j$ and bundle $S$ such that $j\notin S$, $v(j\mid S)\ge (1-c)v(j\mid\emptyset)$.
\end{definition}

The bounded curvature property of submodular functions was introduced by \cite{CC84}, and
has recently received much attention in the optimization and learning
communities (see, e.g., \cite{SDK15} and references within). Sviridenko et al.~\cite{SVW15}
design an approximation algorithm for such valuations with the best-possible approximation factor in
terms of its dependence on the curvature parameter.
We observe the following (the proof appears in Appendix \ref{appx:alpha-submod}): 

\begin{observation}
	\label{obs:curv}
	An $\alpha$-submodular valuation $v$ with curvature $c$ is $\frac{\alpha-1+c}{1-c}$-close to a linear valuation.
\end{observation}

Thus the class of $\alpha$-submodular valuations that are $\epsilon$-close to linear contains as a special case of interest $\alpha$-submodular valuations with bounded curvature.

\subsection{The Greedy Algorithm}
\label{sub:greedy-pos}

Recall that the greedy algorithm (Algorithm \ref{alg:greedy}) is optimal for linear valuations.
In this section we prove the following ``robustness'' theorem for this
algorithm; the implication for welfare maximization is stated in Corollary~\ref{cor:greedy-linear}.

\begin{theorem}
	\label{thm:greedy-linear}
	Let $v$ be an $\alpha$-submodular valuation that is $\epsilon$-close to a linear valuation $\ell$. Let $\mathcal{M}=(M,\cI)$ be a matroid of rank $k$ (represented by an independence oracle). 
	The greedy algorithm returns an independent set $S\in\cI$ such that $v(S)\ge \frac{1-3\epsilon}{\alpha}v(S^*)$, where $S^*\in \arg\max_{I\subseteq \mathcal I} v(I)$. 
\end{theorem}

\begin{corollary}
	\label{cor:greedy-linear}
	In a market with $\alpha$-submodular valuations $v_1,\dots,v_n$ that are $\epsilon$-close to linear valuations $\ell_1,\dots,\ell_n$, there is a polynomial time algorithm with value access that finds an allocation with  $\frac{1-3\epsilon}{\alpha}$-approximately optimal welfare.
\end{corollary}

\begin{proof}
	As in \cite{LLN06}, we define a valuation $v$ over player-item pairs, such that for a set $S$ of such pairs, $v(S)=\sum_i v_i(S_i)$ where $S_i$ is the set of items paired with player $i$ in $S$.  
	We show that $v$ is also $\alpha$-submodular and $\epsilon$-close to linear:
	It is $\epsilon$-close to linear since
	$\sum_i \ell_i(S_i) \le v(S) \le (1+\epsilon)\sum_i \ell_i(S_i)$, and the sum of linear valuations is linear. It is $\alpha$-submodular since for every $S\subseteq T$ and $(i,j)\notin T$, $\alpha v((i,j)\mid S) = \alpha v_i(j\mid S_i)\ge v_i(j\mid T_i) = v((i,j),T)$.
	The proof follows by applying Theorem \ref{thm:greedy-linear} to $v$ and to the partition matroid $\mathcal M$ over player-item pairs, which allows each item to be paired with no more than one player. 
	\qed
\end{proof}

\subsubsection{Proof Idea.} Before we turn to the proof of Theorem \ref{thm:greedy-linear}, we highlight the difference between our analysis in this proof and the standard analysis of the greedy algorithm. Often the analysis of greedy is item by item (see, e.g., \cite{LLN06}), showing that greedily choosing the next item maintains the optimal value up to a small error. However this does not seem to prove a good approximation factor in our case, because of the distortion caused by the pointwise errors. For example, greedy may choose to include an item $j$ from the optimal solution, but at the particular point it is included, its marginal value may not be high relative to its linear contribution (recall the error term in the marginal value as shown in Observation \ref{obs:marginals_close}). This suggests analyzing multiple items at a time in order for the error terms to average out. Indeed, we split the items chosen by greedy into two sets -- those shared with the optimal solution and the rest -- and analyze each set as a whole in order to establish the claimed approximation ratio.

The following lemma (whose proof appears in Appendix \ref{appx:alpha-submod}) relates the sum of marginals to the linear contribution, and is applied in the proof of Theorem \ref{thm:greedy-linear}. Recall that $v$ is $\epsilon$-close to the linear valuation $\ell$, and that for an ordered set of items $X=(x_1,x_2,\dots,x_k)$, $X_j$ denotes the prefix $\{x_1,\dots x_{j-1}\}$. 

\begin{lemma}
	\label{lem:bound}
	Let $X$ and $Z$ be two disjoint sets. $X$ is ordered and has $m_X$ items.  
	Let $Y_1,\dots,Y_{m_X}$ be sets such that $Y_j\subseteq X_j \cup Z$ for every $j$. Then 
	\begin{equation*}
	\alpha\sum_{j=1}^{m_X}{v(x_j\mid Y_j)}\ge \sum_{j=1}^{m_X}{\ell(x_j)}-\epsilon \ell(Z).
	\end{equation*}
\end{lemma}

\begin{proof}[Theorem \ref{thm:greedy-linear}]
	
	Denote by $S$ the independent set that greedy returns given the valuation $v$ and matroid $\mathcal M$, 
	ordered by the order in which the items were greedily added. 
	Denote by $S^*$ an independent set with optimal value $v(S^*)$ subject to the matroid constraint.
	Due to monotonicity we can assume without loss of generality that the sizes of both $S$ and $S^*$ are $k$, i.e., that both sets are bases.
	
	Let $P=\{j\mid s_j\in S\cap S^*\}$ be the positions in $S$ of the items that also appear in $S^*$, and let $Q$ be the remaining positions in $S$.
	Order $S^*$ such that the items in $S\cap S^*$ are in positions $P$. The rest of the items in $S^*$ are ordered among the remaining positions in $S^*$ in the following way: 
	By the exchange property of matroids, there exists a bijection $f$ from $S\setminus S^*$ to $S^*\setminus S$, such that for every position $j\in Q$, $S\setminus\{s_j\}\cup\{f(s_j)\}$ is independent \cite{Sch03}. For every $j\in Q$ we set $s^*_j=f(s_j)$. This ensures that $s^*_j$ can be added to $S_j$ without violating feasibility.    
	We use the notation $S(P)$ and $S(Q)$ for the set of items in $S$ in positions $P$ and $Q$, respectively; similarly for $S^*(P)$ and $S^*(Q)$.
	
	By the definition of greedy, and because adding $s^*_j$ to $S_j$ maintains feasibility for every $j\in Q$, it holds that $v(s_j\mid S_j)\ge v(s^*_j\mid S_j)$ for every $j\in[k]$. 
	Thus 
	\begin{equation}
	v(S) \ge v(\emptyset) + \sum_{j\in P}{v(s_j\mid S_j)} + \sum_{j\in Q}{v(s^*_j\mid S_j)}.\label{eq:pf}
	\end{equation}
		
	We now invoke Lemma \ref{lem:bound} twice. By Lemma \ref{lem:bound} instantiated with $X=S^*(Q)$ and $Z=S$ (note these are indeed disjoint), and using that $S_j\subseteq S$, we get that 
	\begin{equation}
	\alpha\sum_{j\in Q}{v(s^*_j\mid S_j)} \ge \sum_{j\in Q}{\ell(s^*_j)} - \epsilon \ell(S).\label{eq:pf2} 
	\end{equation}
	By Lemma \ref{lem:bound} instantiated with $X=S(P)$ and $Z=S(Q)$ (these are also disjoint), and using that $S_j\subseteq S(Q)\cup \{s_{p_1},\dots,s_{p_{j-1}} \}$, we get that 
	\begin{equation}
	\alpha\sum_{j\in P}{v(s_j\mid S_j)} \ge \sum_{j\in P}{\ell(s_j)} - \epsilon \ell(S) = \sum_{j\in P}{\ell(s^*_j)} - \epsilon \ell(S).\label{eq:pf3}
	\end{equation}
	
	Combining \eqref{eq:pf}, \eqref{eq:pf2}  and \eqref{eq:pf3}, we get that $\alpha v(S)\ge \alpha v(\emptyset) +\sum_{j=1}^k{\ell(s^*_j)} - 2\epsilon \ell(S)$. Since $v$ is $\epsilon$-close to $\ell$ and $S^*$ is optimal for $v$, $\ell(S)\le v(S)\le v(S^*)$. Again using $v$'s closeness to $\ell$, $\alpha v(\emptyset)\ge \alpha c\ge c$. Putting these together we get that $\alpha v(S)\ge \ell(S^*) -2\epsilon v(S^*)\ge v(S^*)/(1+\epsilon)-2\epsilon v(S^*)\ge(1-3\epsilon)v(S^*)$, as required.
	\qed
\end{proof}

% Section 5
\section{The Demand Oracle Model}
\label{sec:results-demand}

\subsection{Rounding the Configuration LP}

By Lemma \ref{lem:LP-val} and Corollary \ref{cor:LP-val} we already
know that \emph{estimation} of the optimal welfare is possible in this
model in polynomial time via the configuration LP. Another question is
whether we can also find an allocation of corresponding value,
given a fractional solution of the configuration LP.  We observe that in some cases,
existing rounding techniques yield
the result that we want. Note that the following proposition does not assume submodularity.

\begin{proposition}
	For combinatorial auctions in the demand oracle model, there is a polynomial time algorithm that finds:
	\begin{itemize}
		\item[$\bullet$] A $(1-\epsilon)$-approximately optimal allocation, if the valuations are
		$\epsilon$-close to linear;
		\item[$\bullet$] A $(1-1/e-\epsilon)$-approximately optimal allocation, if the valuations are
		$\epsilon$-close to XOS (or in particular $\epsilon$-close to submodular).
	\end{itemize}
\end{proposition}

\begin{proof}
	We solve the configuration LP using demand queries. Let $\tilde{v}_i$ denote the true valuations
	and ${v}_i$ the linear/XOS valuations that the $\tilde{v}_i$ are close to. 

	In the case of valuations close to linear, we round the fractional solution by allocating
	each item $j$ to player $i$ with probability $y_{ij} = \sum_{S: j \in S} x_{i,S}$. Call the resulting
	random sets $R_i$. For the underlying linear valuations ${v}_i$, it is clearly the case that
	$\E[{v}_i(R_i)] = \sum_{S} x_{i,S} {v}_i(S)$. Therefore, 
	$$ \sum_i \E[\tilde{v}_i(R_i)] \geq \sum_i \E[v_i(R_i)] =  \sum_{i,S} x_{i,S} {v}_i(S)
	\geq \frac{1}{1+\epsilon} \LP(\tilde{v}_i) \geq (1-\epsilon) \OPT.$$
	
	For valuations close to XOS, we allocate a set $S$ tentatively to player $i$ with probability $x_{i,S}$,
	and then we use the contention resolution technique to resolve conflicts \cite{Fei09,FV10}. 
	This technique has the property that  conditioned on requesting an item $j$, a player receives it with conditional probability at least $1-1/e$.
	Denote by $R_i$ the random set that player $i$ receives after contention resolution.
	Using the fractional subadditivity of XOS functions, we obtain that $\E[v_i(R_i)] \geq (1-1/e) \sum_S x_{i,S} v_i(S)$.
	Hence, similarly to the case above,
	$$ \sum_i \E[\tilde{v}_i(R_i)] \geq \sum_i \E[v_i(R_i)] \geq (1-1/e) \sum_{i,S} x_{i,S} {v}_i(S)
	\geq \frac{1-1/e}{1+\epsilon} \LP(\tilde{v}_i) \geq (1-1/e-\epsilon) \OPT.$$
\end{proof}

\subsection{Positive Results for Kelso-Crawford}
\label{sub:KC-pos}

The Kelso-Crawford algorithm (Algorithm \ref{alg:KC}) uses demand queries in order to find, in each iteration, the items to add to a player's existing bundle.  
In this section we show that it can work well for markets in which the valuations are submodular and $\epsilon$-close to simple subclasses of gross substitutes.
In particular, it finds an allocation and a price vector such that the allocation achieves approximately optimal welfare, and the price vector can be interpreted as approximately stabilizing (in a sense made precise below).  
We can also relax submodularity as long as we maintain approximately decreasing marginals.

\begin{remark}
	For simplicity, our analysis of the Kelso-Crawford algorithm shall treat prices as if raised continuously rather than discretely. This means that the approximation factors we show in this section hold up to a small additive error.\footnote{If the prices are increased by discrete $\delta$ increments, then the additive error is of order $O(mn\delta)$.}
\end{remark}

We begin by establishing that the Kelso-Crawford algorithm works well for the class of $\alpha$-submodular and $\epsilon$-close to linear valuations, as studied above in Section \ref{sec:results-value}. The proof of Theorem~\ref{thm:KC-linear} appears in Section~\ref{sub:KC-proof}.

\begin{theorem}[Linear]
	\label{thm:KC-linear}
	In a market with $\alpha$-submodular valuations that are $\epsilon$-close to linear valuations, the Kelso-Crawford algorithm finds an allocation with $\frac{1}{\alpha+2\epsilon}$-approximately optimal welfare (up to a vanishing error).
\end{theorem}

The positive result in Theorem~\ref{thm:KC-linear} does not extend even to $\epsilon$-close to unit-demand (and still $\alpha$-submodular) valuations -- see Proposition \ref{pro:KC-neg} in Appendix~\ref{appx:KC-neg} for a counter-example. We consider instead the \emph{unweighted} version of unit-demand valuations, and more generally the direct sums of such valuations -- called \emph{unweighted transversal} valuations -- for which we establish that Kelso-Crawford maintains good welfare guarantees. 

A valuation $r$ is an unweighted unit-demand valuation if $v(S)=\max_{j\in S}r(j)$ for every bundle $S$ and $r(j)\in\{0,1\}$ for every item~$j$. 
A valuation $r$ is an unweighted transversal valuation if there exists a partition $\mathcal P=(P_1,\dots,P_k)$ of the items such that $v(S)=\sum_{P\in\mathcal P}\max_{j\in P}r(j)$ for every bundle $S$, and $r(j)\in\{0,1\}$ for every item $j$. An equivalent definition in terms of matchings in graphs is the following: $r$ is an unweighted transversal valuation if it can be represented by a bipartite graph whose vertices on one side of the graph all have degree 1, and whose edges $M$ have $\{0,1\}$ edge weights. The value for a bundle of edges $S$ is the weight of the maximum matching in the bipartite graph induced by $S$. See Figure \ref{fig:transversal} for an example. Such valuations arise in natural economic environments, e.g.~in the context of labor markets.%
\footnote{Consider for example a firm wishing to hire a team of several workers, each with a different specialization, from a pool of specialist workers who are each either acceptable to the firm or not.}

The proof of Theorem \ref{thm:KC-transversal} appears in Section~\ref{sub:KC-proof}.

\begin{figure}
	\caption{A graphical representation of an unweighted transversal valuation $r$. The edges correspond to items and have $\{0,1\}$ weights; they are partitioned in this case into two parts $P_1,P_2$. The value that $r$ attributes to the subset of solid bold edges in this example is 1, since this is the weight of the maximum-weight matching within the subset.}
	\label{fig:transversal}
	\centerline{\includegraphics[width=3in]{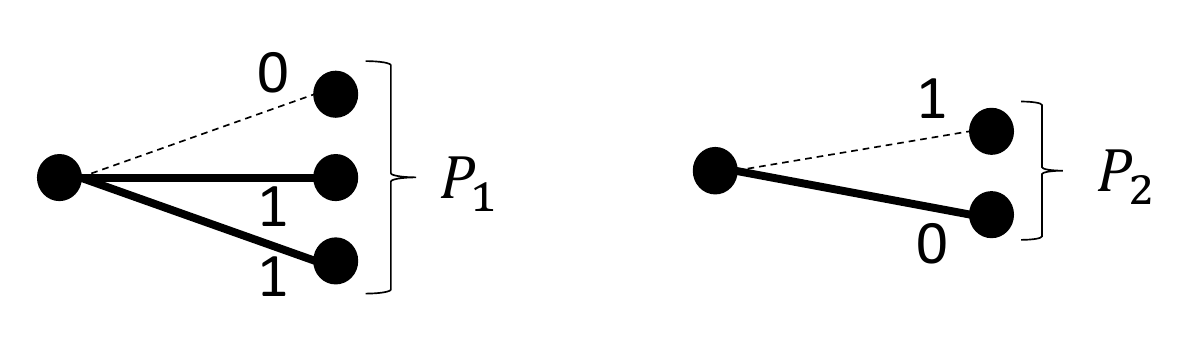}}
\end{figure}

\begin{theorem}[Transversal]
	\label{thm:KC-transversal}
	In a market with $\alpha$-submodular valuations that are $\epsilon$-close to unweighted transversal valuations, the Kelso-Crawford algorithm finds an allocation with $\frac{1}{\alpha(1+3\epsilon)^2}$-approximately optimal welfare (up to a vanishing error).
\end{theorem}

\subsection{Proofs of Positive Results for Kelso-Crawford}
\label{sub:KC-proof}

The proofs of Theorems \ref{thm:KC-linear} and \ref{thm:KC-transversal} follow directly from the welfare guarantee of a ``biased'' Walrasian equilibrium (Definition \ref{def:WE-approx} and Proposition \ref{pro:first-welfare-approx}), together with an appropriate lemma showing that Kelso-Crawford converges to such an equilibrium (the lemma for Theorem \ref{thm:KC-linear} is Lemma \ref{lem:KC-linear} -- see Appendix \ref{appx:KC-linear}; the lemma for Theorem \ref{thm:KC-transversal} is Lemma \ref{lem:KC-transversal} -- see Appendix \ref{appx:KC-transversal}). We first describe the idea of a ``biased'' equilibrium, and then give a high-level overview of how we analyze the performance and convergence of the Kelso-Crawford algorithm. 

\subsubsection{Biased Walrasian equilibrium.}

Recall that the standard proof establishing the optimality (up to a small error) of Kelso-Crawford for gross substitutes relies on showing convergence to a Walrasian equilibrium. 
Unfortunately in our case it does not hold -- even approximately -- that Kelso-Crawford allocates to each player a bundle in his demand. Instead,
our proofs will define and utilize a different notion of an approximate Walrasian equilibrium. 

\begin{definition}
	\label{def:WE-approx}
	Consider a market with $n$ players, and let $\mu\in[0,1]$.
	A full allocation $\mathcal S$ and a price vector $p$ form a \emph{$\mu$-biased Walrasian equilibrium} if there exists $\mu'\in[\mu,1]$ such that for every alternative allocation $T$ and for every player $i$,
	\begin{equation}
	\frac{\mu'}{\mu} v_i(S_i)-p(S_i)\ge \mu' v_i(T_i)-p(T_i).\label{eq:WE-approx}
	\end{equation} 
\end{definition}

\begin{proposition}[Approximate first welfare theorem.]
	\label{pro:first-welfare-approx}
	The welfare of a $\mu$-biased Walrasian equilibrium is a $\mu$-approximation to the optimal welfare.
\end{proposition}

\begin{proof}
	Let $\mathcal{S}^*$ be an optimal allocation. Without loss of generality we can assume that $S^*$ is a full allocation.
	By summing up Inequality \eqref{eq:WE-approx} over all players, we get $(\mu'/\mu)W(\mathcal S)-p(\mathcal S) \ge \mu' W(\mathcal{S}^*)-p(\mathcal{S}^*)$.
	Since both $\mathcal{S},\mathcal{S}^*$ are full allocations, $p(\mathcal{S})=p(\mathcal{S}^*)$. We conclude that $W(\mathcal{S}) \ge \mu\OPT$.
	\qed
\end{proof}

A possible economic interpretation of a $\mu$-biased Walrasian equilibrium is that its allocation and prices induce market stability under the \emph{endowment effect} (see, e.g., \cite{KKT90}). This effect is modeled as an increase of factor $\mu'/\mu$ in a player's value for the bundle he owns, and a decrease of factor $\mu'$ in his value for bundles he does not own.

\subsubsection{Proof Idea}
To give intuition for the proofs of Lemmas \ref{lem:KC-linear} and \ref{lem:KC-transversal} (applied in order to prove Theorems \ref{thm:KC-linear} and \ref{thm:KC-transversal}, respectively),
let us compare the case of linear valuations to that of $\epsilon$-close to linear (and $\alpha$-submodular) valuations. The analysis of Kelso-Crawford for the linear case is simple -- every player ends up with the items for which he has the highest marginal value and can afford to pay the highest price. In the close-to-linear case, an item could end up belonging to the wrong player for two reasons: at some point, given his current allocation (which is subject to changes), either (1) the right player has a low marginal value for the item, or (2) the wrong player has a high marginal value for it. The latter issue is resolved by submodularity, but the former could drive Kelso-Crawford to performance that is bounded away from optimal, as demonstrated in the bad example in Appendix~\ref{appx:KC-neg}. 
Thus the crux of the proofs is to show this cannot  happen for the classes of valuations in question. 

For the close-to-linear case, we achieve this by considering the marginal value of all ``under-valued'' items together, if we were to add them to the player's final allocation (see proof of Lemma \ref{lem:KC-linear}). The transversal case is more involved since we cannot analyze only the addition of items to the player's final allocation -- we need to also consider swapping out items from the final allocation and replacing them with others (see proof of Lemma \ref{lem:KC-transversal}).

% Section 6
\section{Conclusions and Open Questions}
\label{sec:open-questions}

Let us summarize the findings of this paper: The robustness of results for various classes of valuations should not be assumed without further investigation. The default assumption should be that positive results are {\em not} robust and may break down abruptly when small deviations from the class in question are introduced. Robust results can be attained, but they are surprisingly challenging to obtain even in simple cases. Most importantly, it matters how ``closeness" to a class is defined, and what the other attributes of the variant of the problem are (oracle model, class of valuations).

We wish to conclude by stating what we believe to be the main open problems in this area:
\begin{itemize}
	\item[$\bullet$] Is it possible to obtain a $(1-O(\epsilon))$-approximation to optimal social welfare in the value oracle model for submodular valuations that are $\epsilon$-close to gross substitutes?
	\item[$\bullet$] Is it possible to obtain a $(1-O(\epsilon))$-approximation to optimal social welfare in the demand oracle model for (non-submodular) valuations $\epsilon$-close to gross substitutes?
	\item[$\bullet$] Is there any other natural notion of being close to a gross substitutes valuation that captures interesting examples and allows a $(1-O(\epsilon))$-approximation?
\end{itemize}

% Bibliography
\bibliographystyle{splncs03}
\bibliography{abb,bibfile}

% Appendix
\appendix
\section*{APPENDIX}
\setcounter{section}{0}
\section{Standard Algorithms for Welfare-Maximization}
\label{appx:welfare-algos}

Algorithms \ref{alg:KC} and \ref{alg:greedy} describe the ascending price and the greedy approaches to welfare maximization. 

\begin{algorithm}
	\SetAlgoNoLine
	\KwIn{Player valuations $v_1,\dots,v_n$ represented by demand oracles; a parameter $\delta>0$}
	\KwOut{An allocation $\mathcal S$ and a price vector $p$}
	$p := 0$ and $\mathcal S := \emptyset$; \quad\% Initialization
	
	\While{there exists a player $i$ and a non-empty bundle $D_i$ such that $D_i$ is in demand given prices $p+\delta \vec{\mathbbm{1}}_{j\notin S_i}$ and current allocation $S_i$,}{
		$S_i := S_i \cup D_i$;
		
		$S_{i'} := S_{i'}\setminus D_{i}$ for every $i'\ne i$;
		
		$p(j) := p(j) + \delta$ for every $j\in D_i$;
	}
	
	\caption{The Kelso-Crawford ascending-price auction (formulated as an algorithm)}
	\label{alg:KC}
\end{algorithm}

\begin{algorithm}
	\SetAlgoNoLine
	\KwIn{A valuation $v$ represented by a value oracle; a feasibility constraint represented by a feasiblity oracle}
	\KwOut{A bundle $S$}
	$S := \emptyset$; \quad\% Initialization
	
	\While{there exists an item $j\notin S$ such that $S\cup\{j\}$ is feasible}{
		Let $j^*$ be an item that maximizes $v(j^*\mid S)$ among all items $j\notin S$ such that $S\cup\{j\}$ is feasible;
		
		$S := S\cup\{j^*\}$;
	}
	
	\caption{Greedy maximization of value subject to a feasibility constraint}
	\label{alg:greedy}
\end{algorithm}

\section{$\alpha$-Submodularity}
\label{appx:alpha-submod}

\begin{observation}
	\label{obs:alpha-submod}
	If a valuation $v$ is $\alpha$-submodular then it is marginal-$\epsilon$-close for $\epsilon=\alpha-1$ to the class of decreasing functions.
\end{observation}

\begin{proof}[Sketch]
	Fix an item $j$. We want to show that the set function $v(j\mid \cdot)$ is $\epsilon$-close to a decreasing function $g_j$. We define $g_j$ as follows. For every bundle $S$, there is a range $[\frac{1}{1+\epsilon}v(j\mid S),v(j\mid S)]$ to which $g_j(S)$ must belong. Furthermore, for $g_j$ to be decreasing, for every $S$ and every subset $T\subseteq S$, the upper-bound on the range of $g_j(T)$ is also an upper-bound on the range of $g_j(S)$. We claim that: (1) If there is a non-zero range for $g_j(S)$ for every bundle $S$, then by picking the highest value in the range to be $g_j(S)$, we have defined a decreasing $g_j$ such that $v(j\mid \cdot)$ is $\epsilon$-close to it; (2) If there is a bundle $S$ with negative range then we have found a contradiction to $\alpha$-submodularity of $v$. Indeed, if this is the case then we have found a subset $T$ such that $v(j\mid T)<v(j\mid S)/1+\epsilon=v(j\mid S)/\alpha$. This concludes the proof. 
	\qed
\end{proof}

\begin{proof}[of Observation \ref{obs:curv}]
	Define a linear valuation $\ell$ as follows: $\ell(S)=(1-c)(v(\emptyset) + \sum_{j=1}^{|S|}v(s_j\mid \emptyset))$ for every bundle $S$. Using that $v$ has curvature $c$, $v(S) = v(\emptyset) + \sum_{j=1}^{|S|}v(s_j|S_j) \ge (1-c)( v(\emptyset)+\sum_{j=1}^{|S|}v(s_j|\emptyset))=\ell(S)$. Using that $v$ is $\alpha$-submodular, $v(S) = v(\emptyset) + \sum_{j=1}^{|S|}v(s_j|S_j) \le v(\emptyset) + \alpha\sum_{j=1}^{|S|}v(s_j|\emptyset) \le \frac{\alpha}{1-c}\ell(S)$. Thus $v$ is $\epsilon$-close to $\ell$ for $\epsilon$ such that $1+\epsilon=\frac{\alpha}{1-c}$.
	\qed
\end{proof}

\begin{proof}
	\begin{eqnarray}
	\alpha\sum_{j=1}^{m_X}{v(x_j\mid Y_j)} &\ge& \sum_{j=1}^{m_X}{v(x_j\mid Z\cup X_j)}\label{eq:by_submod} \\		
	&=& v(X\mid Z)\label{eq:by_disjoint}\\
	&\le& \ell(X\mid Z) - \epsilon\ell(Z)\label{eq:by_lemma}, 
	\end{eqnarray}
	where \eqref{eq:by_submod} follows since $Y_j\subseteq Z\cup X_j$ and since $v$ is $\alpha$-submodular, \eqref{eq:by_disjoint} follows by the disjointness of $X$ and $Z$, and \eqref{eq:by_lemma} follows since $v$ is $\epsilon$-close to $\ell$ and by Observation \ref{obs:marginals_close}.
	\qed
\end{proof}

\section{Proofs for Kelso-Crawford}
\label{appx:KC}

In this appendix we present several results and missing proofs related to the Kelso-Crawford algorithm. 

\subsection{Kelso-Crawford and $\alpha$-Submodular Valuations}

We generalize a well-known invariant of Kelso-Crawford for submodular valuations \cite{FKL12} to $\alpha$-submodular valuations. This lemma will be used in the next sections to prove Lemmas \ref{lem:KC-linear} and~\ref{lem:KC-transversal}.

\begin{definition}
	Let $v$ be a valuation and $p$ be a price vector. A bundle $S$ is \emph{$\alpha$-IR} ($\alpha$-individually rational) for $v$ given $p$ if $\alpha v(S)\ge p(S)$. A bundle $S$ is \emph{strongly $\alpha$-IR} for $v$ given $p$ if $\alpha v(T)\ge p(T)$ for every $T\subseteq S$. 
\end{definition}

\begin{lemma}
	\label{lem:alpha-strong-IR}
	Consider a market with a player whose valuation $v$ is $\alpha$-submodular. Then the Kelso-Crawford algorithm maintains the following invariant: the player's allocation throughout the algorithm is strongly $\alpha$-IR.
\end{lemma}

\begin{proof}
	Since removing items from the player's allocation maintains the strong $\alpha$-IR property, the only case we need to check is when a bundle $T$ is added to the player's current allocation $S$. We know that $T$ maximizes the player's utility given $S$ and the current price vector $p$. We show this by induction. In the base case all prices are $0$ and so the invariant holds. Now assume for contradiction that $S\cup T$ is not strongly $\alpha$-IR, i.e., there exists a set $S'\cup T'$ where $S'\subseteq S$ and $T'\subseteq T$ such that $\alpha v(S'\cup T')<p(S')+p(T')$. By the induction assumption, $\alpha v(S')\ge p(S')$, and so it must be the case that $\alpha v(T'\mid S')<p(T')$. So using $\alpha$-submodularity we can write the utility $v(T\mid S)-p(T)$ as $v(T\setminus T'\mid S)+v(T'\mid S\cup T\setminus T')-p(T\setminus T')-p(T') \le v(T\setminus T'\mid S)+\alpha v(T'\mid S')-p(T\setminus T')-p(T') < v(T\setminus T'\mid S) -p(T\setminus T')$. Thus adding $T\setminus T'$ to $S$ adds more to the utility than adding $T$ to $S$, contradiction.
	\qed
\end{proof}

A corollary of Lemma \ref{lem:alpha-strong-IR} that generalizes a result of \cite{FKL12} is the following. 

\begin{corollary}
	\label{cor:KC-alpha}
	In a market with $\alpha$-submodular valuations, the Kelso-Crawford algorithm finds an allocation with $(1+\alpha)$-approximately optimal welfare (up to a vanishing error).
\end{corollary}

\subsection{Kelso-Crawford and Close to Linear Valuations}
\label{appx:KC-linear}

\begin{lemma}
	\label{lem:KC-linear}
	In a market with $\alpha$-submodular valuations $v_1,\dots,v_n$ that are $\epsilon$-close to linear valuations $r_1,\dots,r_n$, the Kelso-Crawford algorithm converges to a $\frac{1}{\alpha+2\epsilon}$-biased Walrasian equilibrium.
\end{lemma}

\begin{proof}
	As in the proof of Theorem \ref{thm:KC-transversal}, we can assume without loss of generality that Kelso-Crawford returns a full allocation. 
	From now until the end of the proof, fix a player $i$.  
	Let $S_i$ be player $i$'s allocation and let $p$ be the price vector at termination of the KC algorithm. 
	Let $T_i$ be an alternative bundle for player $i$. We show that 
	\begin{eqnarray}
	v_i(S_i)-p(S_i)\ge v_i(T_i)-p(T_i)-\left(2\epsilon +\alpha - 1\right) v_i(S_i).\label{eq:KC-single-player}	
	\end{eqnarray}
	This is sufficient to complete the proof, since by summing up over all players and rearranging, we get 
	\begin{eqnarray*}
		\left(2\epsilon+\alpha\right)\sum_iv_i(S_i)-\sum_ip(S_i)\ge \sum_i v_i(T_i)-\sum_i p(T_i),
	\end{eqnarray*}
	and so $\mu\ge 1/(\alpha+2\epsilon)$.
	
	It remains to prove Inequality \eqref{eq:KC-single-player}. For simplicity we omit $i$ from the notation. Without loss of generality, assume that in the last round of Kelso-Crawford, the player added to his existing bundle, which we denote by $B$, a bundle which maximizes his utility given $B$ and given the price vector $p$. So the player's utility at termination $v_i(S_i)-p(S_i)$ is at least $u_p(B\cup T)$. We use the notation $B'=B\setminus T$, $T'=T\setminus B$, and $C=B\cap T$. We can now write the lower bound on the player's utility as
	\begin{equation}
	u_p(B\cup T)=v(\emptyset)+v(B'\mid\emptyset)+v(C\mid B')+v(T'\mid B)-p(B')-p(C)-p(T').\label{eq:utility}
	\end{equation}
	By Lemma \ref{lem:alpha-strong-IR}, $B$ is $\alpha$-strongly IR. So 
	\begin{equation}
	v(B'\mid\emptyset) + (\alpha-1) v(B'\mid\emptyset)\ge p(B').\label{eq:IR}
	\end{equation}
	By Observation~\ref{obs:marginals_close}, 
	\begin{equation}
	v(C\mid B')\ge \sum_{j\in C}\ell(j) - \epsilon v(B).\label{eq:app1}
	\end{equation}
	Again by Observation~\ref{obs:marginals_close}, 
	\begin{equation}
	v(T'\mid B) \ge \sum_{j\in T'}\ell(j) - \epsilon v(B).\label{eq:app2}
	\end{equation}
	Plugging Inequalities \eqref{eq:IR}-\eqref{eq:app2} into \eqref{eq:utility} and using $v$'s $\epsilon$-closeness to $\ell$ we get
	\begin{eqnarray*}
		u_p(B\cup T) &\ge& c -(\alpha-1)v(B'\mid \emptyset) + \sum_{j\in T}\ell(j) - p(T) - 2\epsilon v(B)\\
		&\ge & \ell(T)-p(T)- \left(2\epsilon+\alpha-1\right) v(B).
	\end{eqnarray*}
	Inequality \eqref{eq:KC-single-player} follows, and this completes the proof.
	\qed
\end{proof}

\subsection{Kelso-Crawford and Close to Transversal Valuations}
\label{appx:KC-transversal}

\begin{lemma}
	\label{lem:KC-transversal}
	In a market with $\alpha$-submodular valuations $v_1,\dots,v_n$ that are $\epsilon$-close to unweighted transversal valuations $r_1,\dots,r_n$, the Kelso-Crawford algorithm converges to a $\frac{1}{\alpha(1+3\epsilon)^2}$-biased Walrasian equilibrium.
\end{lemma}

\begin{proof}
	By monotonicity we can assume without loss of generality that all items are allocated at price 0 in the first step of the Kelso-Crawford algorithm, and once an item is allocated then it remains allocated throughout. Thus without loss of generality, Kelso-Crawford returns a full allocation as required by the definition of a $\mu$-biased Walrasian equilibrium. 
	
	From now until the end of the proof, fix a player $i$. For simplicity we omit $i$ from the notation. 
	Let $\mathcal P=(P_1,\dots,P_k)$ be the partition of the items corresponding to the unweighted transversal valuation $r$, to which the player's valuation $v$ is $\epsilon$-close. For an item $j$, let $P(j)$ denote the part to which this item belongs. We say that a part $P$ is \emph{represented} in a bundle $X$ if $X$ contains at least one item $j\in P$ with value $r(j)=1$. 
	
	Towards proving the guarantee in Inequality \eqref{eq:WE-approx}, let $S$ be the bundle allocated to the player by the Kelso-Crawford algorithm, and let $T$ be an alternative bundle.  
	Let $P(S)$ (resp., $P(T)$) be the parts represented in $S$ (resp., $T$).  
	Let $S'=\{j\in S\mid P(j)\in P(S)\cap P(T)\}$ be the set of items in $S$ that belong to parts represented both in $S$ and in $T$, and similarly $T'=\{j\in T\mid P(j)\in P(S)\cap P(T)\}$.
	
	\begin{claim*}
		\label{cla:same-rank}
		$r(S')=r(T')$.
	\end{claim*}
	
	\begin{proof}[of Claim \ref{cla:same-rank}]
		The value assigned to a set by valuation $r$ is the number of parts represented in it. By definition, the same parts are represented in $S'$ and in $T'$. 
		\qed
	\end{proof}
	
	We now relate the prices of $S'$ and $T'$ according to the price vector $p$ with which the Kelso-Crawford algorithm terminated. In particular we show that the price of $S'$ cannot be too high in comparison to the price of $T'$.
	
	\begin{claim*} 
		\label{cla:price-relation}
		$p(S')\le p(T') + 3\alpha\epsilon r(S')$. 
	\end{claim*}
	
	The proof of Claim \ref{cla:price-relation} appears below.
	We use the above claims to complete the proof of Lemma \ref{lem:KC-transversal}. 
	Let $S''=S\setminus S'$ and $T''=T\setminus T'$. We know that when the Kelso-Crawford algorithm terminates, the player cannot improve his utility by adding $T''$ to his allocation $S$, and so:
	\begin{eqnarray}
	v(S)-p(S)&\ge& v(S\cup T'')-p(S\cup T'')\nonumber\\
	& =& v(S'')+v(S'\mid S'')+v(T''\mid S)-p(S'')-p(S')-p(T'').\label{eq:transversal-utility}
	\end{eqnarray}
	By $\alpha$-submodularity Lemma \ref{lem:alpha-strong-IR} applies and $S''$ is $\alpha$-IR, therefore
	\begin{equation}
	v(S'')- p(S'')\ge(1-\alpha)v(S'').\label{eq:transversal-IR}
	\end{equation}
	The marginal value assigned by $r$ to a bundle $X$ given a bundle $Y$ is the number of parts represented in $X$ but not in $Y$. Therefore $r(S'\mid S'')=r(S')$ and $r(T''\mid S)=r(T'')$. 
	By Observation~\ref{obs:marginals_close} and since $r(S')=r(T')$ (Claim \ref{cla:same-rank}), 
	\begin{eqnarray}
	v(S'\mid S'') \ge& r(S'\mid S'') - \epsilon r(S'')& \ge r(S') - \epsilon r(S) = r(T') - \epsilon r(S);\label{eq:transversal-app1}\\
	v(T''\mid S) \ge& r(T''\mid S) - \epsilon r(S)&= r(T'') -\epsilon r(S).\label{eq:transversal-app2}
	\end{eqnarray}
	Plugging \eqref{eq:transversal-IR}-\eqref{eq:transversal-app2} into \eqref{eq:transversal-utility}, and using that $r(T')+r(T'')=r(T)$ and $p(S')\le p(T') + 3\alpha\epsilon r(S')$ (Claim \ref{cla:price-relation}),
	\begin{eqnarray*}
		v(S)-p(S) &\ge& (1-\alpha)v(S'') + r(T) - 2\epsilon r(S) - p(T) - 3\alpha\epsilon r(S')\\
		&\ge& \frac{1}{1+\epsilon}v(T)-p(T) - (\alpha-1 + 2\epsilon+3\alpha\epsilon) v(S),
	\end{eqnarray*}
	where the last inequality uses $v$'s $\epsilon$-closeness to $r$. Thus $\mu\ge((1+\epsilon)(\alpha+2\epsilon+3\alpha\epsilon))^{-1}\ge (\alpha(1+3\epsilon)^2)^{-1}$.
	This completes the proof of Lemma \ref{lem:KC-transversal}.
	\qed
\end{proof}

	Towards proving Claim \ref{cla:price-relation}, the following is a property of valuations that are $\alpha$-submodular and $\epsilon$-close to unweighted transversal valuations. It can be seen as a strengthening of Observation \ref{obs:marginals_close}.
	
	\begin{claim*}
		\label{cla:stronger}
		For every part $P$ and bundles $X\subseteq P$ and $Y$, if $P$ is represented in $Y$ then $v(X\mid Y)\le \alpha\epsilon$. 
	\end{claim*} 
	
	\begin{proof}
		Let $y$ be an item representing $P$ in $Y$. Then by $\alpha$-submodularity, $v(X\mid Y)\le \alpha v(X\mid y) \le \alpha r(X\mid y) + \alpha\epsilon r(X\cup\{y\}) = \alpha\epsilon$, where the last inequality is by invoking Observation \ref{obs:marginals_close}. 
		\qed
	\end{proof}

	\begin{proof}[of Claim \ref{cla:price-relation}]
		Fix a part $P$ that is represented in $S'$ and $T'$. 
		Let $t$ be an item representing $P$ in $T'$. 
		Order the items representing $P$ in $S'$ according to the order in which the Kelso-Crawford algorithm added them to the player's allocation for the last time (i.e., at their termination prices according to $p$), breaking ties arbitrarily; denote the ordered set by $(s_1,s_2,\dots)$. Let $B_j$ be the bundle of the player right before item $s_j$ was added, and let $D_j$ be the set of items (not including $s_j$) with which $s_j$ was added to~$B_j$.  
		
		We first consider item $s_1$. If it is the case that $P$ is not represented in $B_1$ and no other item represents $P$ in $D_1$, then we argue that 
		$v(s_1\mid B_1\cup D_1)-p(s_1)\ge v(t\mid B_1\cup D_1)-p(t)$. The reason for this is that otherwise, the player's utility could have been improved by replacing $s_1$ by $t$ (using that $t$'s termination price $p(t)$ is weakly higher than its price when $s_1$ was added). By monotonicity, we can write $v(s_1\mid B_1\cup D_1) \le v(s_1,t\mid B_1\cup D_1)\le v(t\mid B_1\cup D_1)+\alpha\epsilon$, where the last inequality is by Claim \ref{cla:stronger}. Therefore $p(s_1)\le p(t)+\alpha\epsilon$.
		In the remaining case, $P$ is represented in either $B_1$ or $D_1$. We know that $v(s_1\mid B_1\cup D_1)\ge p(s_1)$, otherwise the utility could have been improved by dropping~$s_1$. By Claim \ref{cla:stronger}, the left-hand side is at most $\alpha\epsilon$, and so $p(s_1)\le \alpha\epsilon$.
		
		Now consider the rest of the items $s_2,s_3,\dots$. When item $s_j$ is added, $\{s_1,\dots,s_{j-1}\}\subseteq B_j\cup D_j$. We also know that, as above, $v(s_j\mid B_j\cup D_j)\ge p(s_j)$. By $\alpha$-submodularity this means that $\alpha \sum_{j\ge 2} v(s_j\mid S_j) \ge \sum_{j\ge 2} p(s_j)$. The left-hand side is equal to $\alpha v(\{s_2,s_3,\dots\}\mid s_1)$, and is therefore $\le \alpha\epsilon$ by Observation \ref{obs:marginals_close}. We have thus shown that $\sum_{j\ge 2} p(s_j)\le \alpha\epsilon$, and the same argument shows that the total payment for items in $S'\cap P$ that do not represent $P$ is at most $\alpha\epsilon$. 
		
		We conclude that the total payment for items in $S'\cap P$ is at most $p(t)+ 3\alpha\epsilon$. Summing up over all parts represented in $S'$ and $T'$ completes the proof of the claim.
		\qed
	\end{proof}

\section{Negative Results for Specific Algorithms}
\label{appx:negative}

How well do standard algorithms for welfare maximization perform for valuations that are close to, but are not quite, gross substitutes? 
After discussing the LP-based approach in Section \ref{sub:value-vs-demand},
in this section we discuss the  two other main approaches to welfare maximization for gross substitutes -- the ascending auction algorithm of Kelso and Crawford \cite{KelsoCrawford1982}, and the cycle canceling algorithm of Murota \cite{Mur96a,Mur96b}. 

At first glance the Kelso-Crawford algorithm seems like a promising approach due to its known welfare guarantee of a $1/2$-approximation for the class of submodular valuations \cite{FKL12}. However
Section~\ref{appx:KC-neg} shows that the Kelso-Crawford algorithm cannot in general guarantee much better than that even for simple submodular valuations that are arbitrarily close to unit-demand. An interesting open question is whether the Kelso-Crawford algorithm can be modified to eliminate such bad examples, e.g., by ordering the players in an optimal way.

The cycle canceling algorithm of Murota is a different approach which relies on properties of GS (or equivalently $M^\natural$-concave) valuations under local improvements. In Section \ref{sub:murota-neg} we show that such properties hold for submodular valuations in a certain approximate sense, but unfortunately the cycle canceling algorithm cannot find a local optimum that would allow us to exploit these properties. In fact we show that a local optimum in Murota's sense does not exist for submodular valuations.

\subsection{A Negative Result for Kelso-Crawford}
\label{appx:KC-neg}

\begin{proposition}
\label{pro:KC-neg}
	For every $\epsilon\le 1$, there exists a market with $O(1/\epsilon)$ players whose submodular valuations are $\epsilon$-close to unit-demand, for which the Kelso-Crawford algorithm with adversarial ordering of the players finds an allocation with $\approx2/3$ of the optimal welfare.
\end{proposition}

\begin{proof}
	Let $\rho=\frac{\epsilon}{1+\epsilon}$.
	Let $H$ be a large constant, and set $n'=H/\rho$. Let $\delta$ be negligibly small. 
	
	There are three kinds of items on the market: one ``large'' item $h$, $n'$ items $x_1,\dots,x_{n'}$, and $n'$ items $y_1,\dots,y_{n'}$.
	There are four kinds of players: $n'$ main players, two $h$-players who have unit-demand valuations for the item $h$, $2n'$ $x$-players who have unit-demand valuations for one of the $x$'s, and $2n'$ $y$-players who have unit-demand valuations for one of the $y$'s. 
	In particular, the $h$-players value item $h$ at $H$, each $x_j$ has a pair of $x$-players who value it at $\rho H-\delta$, and each $y_j$ has a pair of $y$-players who value it at $\rho H$. 
	Denote by $v_i$ the valuation of the $i$th main player, and by $r_i$ the unit-demand valuation to which it is $\epsilon$-close. Table \ref{tab:approx-unit-demand} compares $v_i$ and $r_i$ for the items to which they attribute non-zero value (all other items have value $0$). The proof of the next claim appears below.
	
	\begin{table}
		\caption{The valuation $v_i$ of the $i$th main player, and the unit-demand valuation $r_i$ to which it is $\epsilon$-close.}
		\centering
			\begin{tabular}{ | c || c | c | c | c | c | c | c | c | }
				\hline
				&&&&&&&&\\
				& ~$\{\}$~ & ~$\{h\}$~ & ~$\{x_i\}$~ & ~$\{y_i\}$~ & $\{h,x_i\}$ & $\{h,y_i\}$ & $\{x_i,y_i\}$ & $\{h,x_i,y_i\}$ \\
				&&&&&&&&\\
				\hline
				\hline
				&&&&&&&&\\
				~$v_i$~ & 0 & $\frac{H}{1+\epsilon}$ & $\rho H$ & $2\rho H$ & $H$ & $H$ & $2\rho H$ & $H$ \\
				&&&&&&&&\\
				\hline
				&&&&&&&&\\
				$r_i$ & 0 & $\frac{H}{1+\epsilon}$ & $\frac{\rho H}{1+\epsilon}$ & $\frac{2\rho H}{1+\epsilon}$ & $\frac{H}{1+\epsilon}$ & $\frac{H}{1+\epsilon}$ & $\frac{2\rho H}{1+\epsilon}$ & $\frac{H}{1+\epsilon}$\\
				&&&&&&&&\\  
				\hline
			\end{tabular}
			\label{tab:approx-unit-demand}
		\end{table} 
		
		\begin{claim*}
			\label{cla:KC-neg-submod}
			The valuation $v_i$ is submodular. 
		\end{claim*}
		
		We now describe how the Kelso-Crawford algorithm can go wrong. First, the $x$-players and $y$-players raise the price of every item $x_j$ to $\rho H-\delta$ and of every item $y_j$ to $\rho H$. The price of $h$ remains $0$ at this point. It is now the turn of the first main player. His utilities for the relevant bundles are described in Table \ref{tab:possible-utilities}. The first main player chooses $\{h,x_1\}$ as the unique bundle that maximizes his utility, and the prices of items $h$ and $x_1$ slightly increase. Now the second main player chooses $\{h,x_2\}$ as the bundle that maximizes his utility, taking $h$ away from the first main player (the price of $h$ has increased, but if the price increments are small enough -- of order $O(\epsilon)$ -- then it is still the unique utility-maximizing bundle for the second main player). This repeats for the third main player and so on. After all the main players have chosen bundles in this way, the $h$-players raise the price of item $h$ to $H$. At this point, no player wants to add to his current bundle: For the main players, $h$ is now too expensive, and since for every $i$ the $i$th main player has $\{x_i\}$ as his bundle, the marginal value of $y_i$ for him is equal to $y_i$'s price $\rho H$. For the $x$-players, the $x$ items have become too expensive. For the $y$-players and the $h$-players, the $y$ and $h$ items have prices equal to their values, respectively. Thus the Kelso-Crawford algorithm terminates.
		
		\begin{table}
			\caption{The $1$st main player's utilities for different bundles after the prices of the $x$ and $y$ items were raised (the raised prices appear in the first row).}
			\centering
				\begin{tabular}{ | c || c | c | c | c | c | c | c | c | }
					\hline
					&&&&&&&&\\
					& ~$\{\}$~ & ~$\{h\}$~ & ~$\{x_1\}$~ & ~$\{y_1\}$~ & $\{h,x_1\}$ & $\{h,y_1\}$ & $\{x_1,y_1\}$ & $\{h,x_1,y_1\}$\\
					&&&&&&&&\\
					\hline
					\hline
					&&&&&&&&\\
					$~p~$ & 0 & 0 & $\rho H-\delta$ & $\rho H$ & $\rho H-\delta$ & $\rho H$ & $2\rho H-\delta$ & $2\rho H-\delta$\\
					&&&&&&&&\\
					\hline
					&&&&&&&&\\
					$v_1 - p$ & 0 & $\frac{H}{1+\epsilon}$ & $\delta$ & $\frac{\epsilon H}{1+\epsilon}$ & $\frac{H}{1+\epsilon}+\delta$ & $\frac{H}{1+\epsilon}$ & $\delta$ & $\frac{(1-\epsilon)	H}{1+\epsilon} + \delta$\\
					&&&&&&&&\\
					\hline
				\end{tabular}
				\label{tab:possible-utilities}
			\end{table} 
			
			Consider the welfare at termination. The main players contribute $n'\rho H$ to the welfare, the $y$-players also contribute $n'\rho H$ and $h$-players contribute $H$. 
			Alternatively, if the main players were to take the $y$ items, the $x$-players were to take the $x$ items and the $h$-players would keep the $h$ item, the welfare would have been
			$2n'\rho H + n'\rho H + H$ and this is optimal.
			The ratio between the welfare achieved by Kelso-Crawford and $\OPT$ is therefore $\approx \frac{3}{4}$, and this completes the proof.
			\qed
		\end{proof}

\begin{proof}[of Claim \ref{cla:KC-neg-submod}]
	The valuation $v_i$ attributes non-zero values to three items: $h$, $x_i$ and $y_i$. 
	Table \ref{tab:approx-unit-demand} shows $v_i$'s monotonicity. To see its submodularity, consider
	Table \ref{tab:KC-neg-marginals}, which shows the marginals of the three items (one row per item) given combinations of the other two items. Note that adding any other items on the market does not change these marginals. We can see that the marginals for every item are indeed decreasing, concluding the proof.
\begin{table}
	\caption{The marginal values of valuation $v_i$ for items $h,x_i,y_i$.}
	\centering
		\begin{tabular}{ | c || c | c | c | c | }
			\hline
			&&&&\\
			$~h~$ & $v_i(h)=$ & $v_i(h\mid x_i)=$ & $v_i(h\mid y_i)=$ & $v_i(h \mid x_i,y_i)=$\\
			& $\frac{H}{1+\epsilon}=(1-\rho)H$ & $(1-\rho)H$ & $(1-2\rho)H$ & $(1-2\rho)H$ \\
			&&&&\\
			\hline
			&&&&\\
			$x_i$ & $v_i(x_i)=$ & $v_i(x_i\mid h)=$ & $v_i(x_i\mid y_i)=$ & $v_i(x_i \mid h,y_i)=$\\
			& $\rho H$ & $\rho H$ & 0 & 0 \\
			&&&&\\
			\hline
			&&&&\\
			$y_i$ & $v_i(y_i)=$ & $v_i(y_i\mid h)=$ & $v_i(y_i\mid x_i)=$ & $v_i(y_i \mid h,x_i)=$\\
			& $2\rho H$ & $\rho H$ & $\rho H$ & 0 \\
			&&&&\\
			\hline
		\end{tabular}
	\label{tab:KC-neg-marginals}
\end{table}
\qed
\end{proof}

\subsection{Murota's Cycle Canceling Approach}
\label{sub:murota-neg}

The cycle canceling approach is based on the following useful property of gross substitutes valuations, called the \emph{single improvement (SI)} property \cite{GS99}: 
Given a price vector $p$, we say that a bundle $S$ is in \emph{local demand} for a valuation $v$ if its utility cannot be improved by adding an item, removing an item or swapping one item for another. The SI property holds if every bundle $S$ in local demand is also in (global) demand (i.e., $v(S)-p(S)\ge v(T)-p(T)$ for every bundle $T$). Murota's cycle canceling algorithm finds an allocation and prices such that each player gets a bundle in local demand, thus arriving at an optimal allocation for gross substitutes valuations. 

We observe that a variant of the SI property characterizes submodular valuations, and thus the SI property interpolates between submodularity and gross substitutes.

\begin{definition} 
	Let $\beta\in[0,1]$. A valuation $v$ is $\beta$-SI if for every $S$ in local demand, $S$ is strongly individually rational,%
	\footnote{Being strongly IR implicitly holds for SI as well, i.e., if $v$ is SI and $S$ is in local demand then $S$ is strongly IR.}
	and $v(S)-\beta p(S)\ge v(T)-p(T)$ for every $T$.
\end{definition}

\begin{observation}
	\label{obs:neg-submod-equiv-SI}
	A monotone valuation is submodular if and only if it is $0$-SI.
	A monotone valuation is gross substitutes if and only if it is $1$-SI.
\end{observation}

\begin{proof}
Suppose $v$ is submodular; we want to prove that $v$ is $0$-SI. Let $S$ be in local demand, $v(S+i) - v(S) \leq p_i$, and $v(S-j) - v(S) \leq -p_j$. (We don't even need the swap property.)
By submodularity and the second property, we have $v(S \cup T) \leq v(S) + \sum_{i \in T \setminus S} p_i$.
Therefore, $v(S) \geq v(S \cup T) - p(T \setminus S) \geq v(T) - p(T)$ by monotonicity. Also, for every $S' \subset S$, by submodularity we have $v(S \setminus S') \leq v(S) - \sum_{j \in S'} p_j$. Therefore $v(S') \geq v(S) - v(S \setminus S') \geq p(S')$.

Conversely, suppose that $v$ is not submodular, i.e.~$v(S+a+b) > v(S) + v(a \mid S) + v(b \mid S)$ for some $S$ and $a,b \notin S$. We set $p_i = 0$ for $i \in S$ and $p_j = v(j \mid S)$ for $j \notin S$. Clearly $S$ is in local demand, because swapping at most 1 element does not increase utility. However, $v(S \cup \{a,b\}) - p(S \cup \{a,b\}) > v(S)$, so $v$ is not $0$-SI.
\qed
\end{proof}

Moreover, as the next observation shows, finding an allocation with prices such that each player gets a bundle in their local demand would give a smooth transition would provide a smooth transition between a $1/2$-approximation for submodular valuations and an optimal solution for gross substitutes.

\begin{observation}
\label{obs:neg-1-welfare}
For $\beta$-SI valuations, any allocation (of all items) with prices such that each player gets a bundle in their local demand has value at least $\frac{1}{2-\beta} OPT$.
\end{observation}

\begin{proof}
Suppose that each $v_i$ is $\beta$-SI and we have an allocation $(A_1,\ldots,A_n)$ with prices such that each $A_i$ is in the local demand of player $i$. We assume that all items are allocated, and the allocation is individually rational; therefore, $\sum_{i=1}^{n} v_i(A_i) \geq \sum_{i=1}^{n} p(A_i) = \sum p_j$. 

Suppose that the optimal allocation is $(O_1,\ldots,O_n)$. By the $\beta$-SI property, we have $v_i(A_i) - \beta p(A_i) \geq v_i(O_i) - p(O_i)$. Adding up over all players, $\sum_{i=1}^{n} v_i(A_i) - \beta \sum p_j \geq OPT - \sum p_j$. From here, $\sum_{i=1}^{n} v_i(A_i) \geq OPT - (1-\beta) \sum p_j \geq OPT - (1-\beta) \sum_{i=1}^{n} v_i(A_i)$. Thus $(2-\beta) \sum_{i=1}^{n} v_i(A_i) \geq OPT$.
\qed
\end{proof}

Unfortunately, this approach (whether by cycle canceling or by any other method) is doomed to fail: Example \ref{ex:neg-no-local} demonstrates that there are markets with submodular valuations for which it is impossible to find an allocation and prices such that every player's bundle is in their local demand. In other words, it is impossible to get rid of all negative cycles in Murota's algorithm. 

\begin{example}
\label{ex:neg-no-local}
We consider 2 players, 4 items, and submodular (coverage) valuations defined as follows:
$v_1(S) = w(\bigcup_{i \in S} A_i)$, where the weights of the regions in the Venn diagram of $A_1,A_2,A_3,A_4$ are depicted in Figure~\ref{fig:coverage}; $\epsilon>0$ here is a small constant. We also define $v_2(S) = w(\bigcup_{i \in S} B_i)$ where $B_1 = A_4$, $B_2 = A_1$, $B_3 = A_2$, $B_4 = A_3$.

The claim is that the optimal allocation is $S_1 = \{1,3\}$ for player $1$, and $S_2 = \{2,4\}$ for player $2$, which has welfare $16$.
However, given this allocation there is a negative cycle in Murota's sense:
\begin{itemize}
\item[$\bullet$] Player $1$ swaps item $1$ for item $2$;
\item[$\bullet$] Player $2$ swaps item $2$ for item $3$;
\item[$\bullet$] Player $1$ swaps item $3$ for item $4$;
\item[$\bullet$] Player $2$ swaps item $4$ for item $1$.
\end{itemize}
Each exchange alone brings an improvement in value, from $8$ to $8+\epsilon$ for each player, and thus is represented by a negative edge in Murota's graph. However, performing all the exchanges in a cycle decreases the value from $16$ to $8+4\epsilon$. 

For other allocations, it is even easier to see that there is a negative cycle: The reader can check that for every other allocation (which is suboptimal), there is a swap of items that improves the solution. Hence this swap corresponds to a negative cycle. 

It follows that for every allocation there is a negative cycle in Murota's graph. Equivalently, for any assignment of prices, there exists a local swap that improves some player's utility. (Prices without an improving swap exist if and only if there is no negative cycle in Murota's graph; see \cite{Pae14}.) In other words, we cannot arrange for all players to have a bundle in their local demand.

One might object that these coverage valuations are not close to gross substitutes. However, one can easily achieve this by adding
a linear function with large coefficients; the resulting valuation is close to linear, and the same conclusion still holds --- we cannot get rid of negative cycles. 

We remark that \cite{MM15} discuss the issue of non-improving negative cycles, and present a heuristic modification of their algorithm for submodular valuations. However, as far as we see, they do not address the question whether a solution without negative cycles could be found by some other approach. Our example shows that it is indeed impossible to eliminate negative cycles for submodular valuations.

\begin{figure}
\caption{Coverage valuation for a counterexample to cycle canceling.}
\label{fig:coverage}
\begin{tikzpicture}

\node at (-3,8) {};

\draw (2,2) circle (2cm);
\draw (5,2) circle (2cm);
\draw (2,5) circle (2cm);
\draw (5,5) circle (2cm);

\node at (1.5,1.5) {$\epsilon$};
\node at (5.5,1.5) {$\epsilon$};
\node at (1.5,5.5) {$2$};
\node at (5.5,5.5) {$2$};

\node at (3.5,2) {$2$};
\node at (3.5,5) {$2$};
\node at (2,3.5) {$1$};
\node at (5,3.5) {$1$};

\node at (0,6.5) {$A_1$};
\node at (0,0.5) {$A_2$};
\node at (7,6.5) {$A_3$};
\node at (7,0.5) {$A_4$};

\end{tikzpicture}
\end{figure}
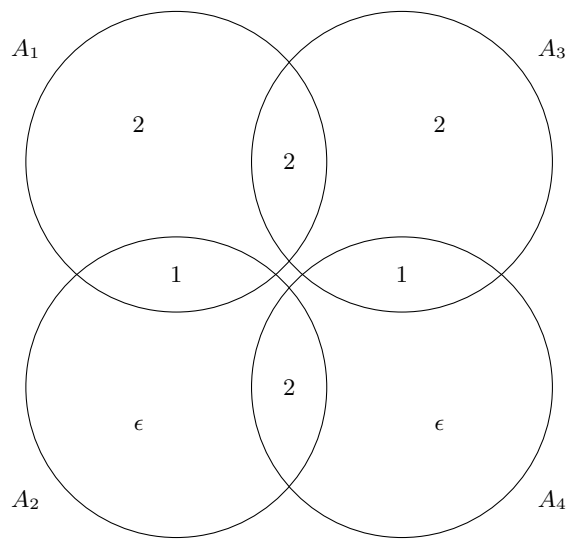

\end{example}

\end{document}